\documentclass[reqno,a4paper,10pt]{amsart}

\usepackage[foot]{amsaddr}
\usepackage{booktabs} % for nicer looking tables
\usepackage{graphicx,multirow,hhline}
\usepackage{amsmath,amsfonts,amssymb}
\usepackage[hidelinks]{hyperref}
\newtheorem{thm}{Theorem}
\newtheorem{cor}[thm]{Corollary}
\newtheorem{lemma}[thm]{Lemma}
\newtheorem{prop}[thm]{Proposition}
\theoremstyle{definition}

\theoremstyle{remark}
\newtheorem{rem}{Remark}
\newtheorem{example}{Example}
\newtheorem{notation}{Notation}

\newcommand{\set}[1]{\left\{#1\right\}}

\newcommand{\PB}{\left\{\cdot\,,\cdot\right\}}
\newcommand{\pb}[1]{\left\{#1\right\}}
\newcommand{\LVPS}{\operatorname{LVPS}}

\begin{document}

\parskip 4pt
\baselineskip 16pt

\title[Poisson structures for difference equations]{Poisson structures for difference equations}

\author{C. A. Evripidou\textsuperscript{1}, G. R. W. Quispel\textsuperscript{1}}
\address{\textsuperscript{1}Department of Mathematics and Statistics, La Trobe University, Melbourne, Victoria 3086, Australia}
\email{C.Evripidou@latrobe.edu.au, R.Quispel@latrobe.edu.au}

\author{J. A. G. Roberts\textsuperscript{2}}
\address{\textsuperscript{2}School of Mathematics and Statistics UNSW Australia, Sydney NSW 2052, Australia}
\email{Jag.Roberts@unsw.edu.au}

\date{\today}

\begin{abstract}
We study the existence of log-canonical Poisson structures that
are preserved by difference equations of special form. We also
study the inverse problem, given a log-canonical Poisson structure
to find a difference equation preserving this structure. We give
examples of quadratic Poisson structures that arise for the
Kadomtsev-Petviashvili (KP) type maps which follow from a
travelling-wave reduction of the corresponding integrable
partial difference equation.
\end{abstract}

\maketitle

\section{Introduction}

Hamiltonian dynamical systems form a major area of study in dynamical
systems, both for their mathematical structure and because of their widespread applications
\cite{abrmar,goldstein}.
The form of the paradigmatic Hamiltonian system is
\begin{equation}
\label{eq:sympl_str}
\dot{\bf x} = \Omega \; \nabla H({\bf x})
\end{equation}
where ${\bf x} \in \mathbb R^n$ with $n$ even. The $n \times n$
constant matrix $\Omega$ is skew-symmetric and is the \textit{symplectic
structure} of the system. Typically,
\begin{equation}
\label{eq:sympl_matr}
\Omega=
\begin{pmatrix}
0& I\\
- I& 0
\end{pmatrix}
\end{equation}
where $0$ and $I$ denote, respectively, the zero and identity
matrix of dimension $\frac{n}{2}$. More generally, $\Omega$ in
\eqref{eq:sympl_str} can be any constant skew-symmetric matrix,
or more generally again, a non-constant skew symmetric matrix $\Omega({\bf x})$
which satisfies the Jacobi identity, in which case it is called a \textit{Poisson structure}.
The Poisson structure is called non-degenerate when
$$
\det(\Omega({\bf x}))\neq 0.
$$
These possibilities for $\Omega({\bf x})$ can also be taken for \eqref{eq:sympl_str}
in the case that the dimension is odd in which
case $\Omega({\bf x})$ is degenerate because it is skew-symmetric and odd dimensional.

Equations \eqref{eq:sympl_str} can be written in terms of the Poisson bracket $\PB$
defined by $\pb{x_i,x_j}({\bf x}):=\Omega_{i,j}({\bf x})$ or, for any functions $f, g$, by
\begin{equation}
\label{eq:pb_rel}
\pb{f({\bf x}),g({\bf x})}({\bf x}):=\sum_{1\leq i,j \leq n}\frac{\partial f}{\partial x_i}
({\bf x})\frac{\partial g}{\partial x_j}({\bf x})\pb{x_i,x_j}({\bf x})=\nabla f({\bf x})^t\Omega({\bf x})
\nabla g({\bf x}).
\end{equation}
Then equations \eqref{eq:sympl_str} become
$$
\dot{x_i}=\pb{x_i,H}({\bf x}).
$$
The existence of a symplectic structure or more generally of a Poisson
structure plays a key role in the geometry of \eqref{eq:sympl_str}.
Darboux's theorem says that any system \eqref{eq:sympl_str} with arbitrary
non-degenerate Poisson matrix $\Omega({\bf x})$ can be transformed locally to the
Hamiltonian form with the canonical $\Omega$ of \eqref{eq:sympl_matr}.
However, many systems arise naturally with a non-canonical $\Omega({\bf x})$
and are analyzed in that given coordinate system (e.g. for geometric numerical
integration, where the system is numerically approximated with the symplectic
structure not converted to canonical form).

In discrete time, a map
\begin{equation}
\label{eq:map}
M : {\bf x} \mapsto {\bf x'}:=M({\bf x})
\end{equation}
preserves a Poisson structure $\Omega({\bf x})$ if its Jacobian matrix $dM(\bf{x})$, satisfies 
$$
dM^t({\bf x})\; \Omega({\bf x})\; dM ({\bf x}) = \Omega({\bf x}').
$$
Equivalently if for any two functions $f,g$ on $\mathbb R^n$,
\begin{equation}
\label{eq:pois_pres}
\pb{f\circ M,g\circ M}({\bf x})=\pb{f,g}(M({\bf x}))
\end{equation}
which, by using the notation $G\circ M({\bf x})=G'$ for any function $G$,
can be shortened into $\pb{f',g'}=\pb{f,g}'$. Using \eqref{eq:pb_rel}
this is equivalent to $\pb{x_i',x_j'}=\pb{x_i,x_j}'$ for all $1\leq i< j \leq n$.
As in the continuous case, the existence of the Poisson structure
plays a key role in the geometry of \eqref{eq:map} (see \cite{ves}).
Recall that the flow
or map in $2m$ degrees of freedom satisfies Liouville-Arnol'd integrability if
there exist $m$ functionally independent integrals of motion
$\set{I_1 , I_2 , \ldots , I_m }$, in involution with respect to the Poisson
structure, i.e. satisfying $\pb{I_i,I_j}=0$. Clearly, establishing this type of
integrability requires knowing the Poisson structure to begin with.

In this paper we study the problem of finding a Poisson structure $\PB$
that is preserved by a difference equation of order $n$ of the form
\begin{equation}
\label{eq:diff_equat}
x_{n}=F({\bf x}):=F(x_0,x_1,\ldots,x_{n-1}).
\end{equation}
By saying that the Poisson structure $\PB$ is preserved by the
difference equation \eqref{eq:diff_equat} we mean that the map
\begin{equation}
\label{eq:map_gener}
M(x_0,x_1,\ldots,x_{n-1})=(x_0',x_1',\ldots,x_{n-1}')=:\bf x'
\end{equation}
where
\begin{align}
\begin{split}
\label{eq:def_of_M}
x_i'=x_{i+1}\quad \text{ for } i=0,1,\ldots, n-1,\quad \text{and}\quad
x_n=F({\bf x})
\end{split}
\end{align}
is a Poisson map. By now, many authors have studied similar problems from
the point of view of cluster algebras \cite{hone_1, hone_2}, r-matrix approach
\cite{kouloukas_1}, using three leg forms for $(p,p)$ reductions of maps in the
ABS list \cite{ABS}, by considering symplectic structures \cite{iatrou}
and many other \cite{brst,papa_1, emmrich, faddeev, dihn_4, dihn_1, papa_2, dihn_3, dihn_2}.

We will consider two families of difference equations of the form
\eqref{eq:diff_equat} with
\begin{equation}
\label{eq:first_form}
F({\bf x})=
%\phi({\bf r}_1\cdot{\bf x},{\bf r}_2\cdot{\bf x},\ldots,{\bf r}_k\cdot{\bf x})=
\phi(y_1,y_2,\ldots,y_k)
\end{equation}
where
$$
y_i={\bf r}_i\cdot{\bf x}=\sum_{j=0}^{n-1}r_{i,j}x_j
$$
is the dot product of ${\bf x}$ and
${\bf r}_i = (r_{i,0},r_{i,1},\ldots, r_{i,n-1}) \in \mathbb R^n$ and
\begin{equation}
\label{eq:second_form}
F({\bf x})=
%\psi({\bf x}^{{\bf r}_1},{\bf x}^{{\bf r}_2},\ldots,{\bf x}^{{\bf r}_k})=
\psi(z_1,z_2,\ldots,z_k)
\end{equation}
where
$$
z_i={\bf x}^{{\bf r}_i}=x_0^{r_{i,0}}x_1^{r_{i,1}}\cdots x_{n-1}^{r_{i,n-1}},
$$
for functions $\phi, \psi\in C^1(\mathbb R^k)$.
Some particular choices of the functions $\phi$ and $\psi$ give rise to several well
known maps such as the Sine-Gordon (SG), Korteweg-de Vries (KdV) ,
modified KdV (mKdV), potential KdV (pKdV), AKP and BKP reductions \cite{hone_3, hone_2, dihn_4, papa_2, dihn_3, dihn_2}
as Table \ref{tbl:table1} shows. At the end of Section
\ref{sec:map_to_poisson} we study the maps presented in
Table \ref{tbl:table1} and we provide Poisson structures
that are preserved by them. For simplicity we adopt the following notation.

\begin{notation}
\label{not:bar}
The bar over a sequence of numbers means that the sequence
is repeated and the number of times is repeated will follow
from the order of the corresponding map.
For the KdV map in Table \ref{tbl:table1}
below the vector ${\bf r}_1$ contains only $-1$'s and the
vector ${\bf r}_2$ has zero in its first two and the last
element and the rest are $1$'s.
\end{notation}

\begin{center}
\begin{table}[ht]
  \caption{Special cases of the function $F$ that are related to known maps}
  \label{tbl:table1}
  \begin{tabular}{ | c | c | }
    \hline
     Form of the function $F$ & Related map \\ \hhline{|=|=|}
      \multirow{2}{*}{ \quad \; $F({\bf x})=\phi(y_1,y_2)={\bf r}_1\cdot
      {\bf x}+\frac{p_1{\bf r}_2\cdot{\bf x}+q_1}{p_2{\bf r}_2\cdot{\bf x}+q_2}$ \quad \; } &
      \hspace{-2ex} \;  KdV: \hspace{2ex} ${\bf r}_1 =(\overline{-1})$ \\ &
      \hspace{8.5ex} ${\bf r}_2 =(0,0,\overline{1},0)$                      \\ \hline
      \multirow{2}{*}{ \quad \; $F({\bf x})=\phi(y_1,y_2)={\bf r}_1\cdot{\bf x}+
      \frac{p}{{\bf r}_2\cdot{\bf x}}$ \quad \; } &
      \quad \; \hspace{-5ex} pKdV: \quad ${\bf r}_1 =(1,\overline{0})$   \\ &
      \hspace{10.5ex} ${\bf r}_2 =(0,-1,\overline{0},1)$                     \\ \hline
      \multirow{2}{*}{ \quad \; $F({\bf x})=\psi(z_1,z_2)={\bf x}^{{\bf r}_1}
      \frac{p_1{\bf x}^{{\bf r}_2}+q_1}{p_2{\bf x}^{{\bf r}_2}+q_2}$ \quad \; } &
      \quad \; \hspace{-5ex} mKdV: \hspace{2ex} ${\bf r}_1 =(1,\overline{0})$ \\ &
      \hspace{10.5ex} ${\bf r}_2 =(0,-1,\overline{0},1)$                          \\ \hline
      \multirow{2}{*}{ \quad \; $F({\bf x})=\psi(z_1,z_2)=
      {\bf x}^{{\bf r}_1}\frac{p_1{\bf x}^{{\bf r}_2}+q_1}{p_2{\bf x}^{{\bf r}_2}+q_2}$ \quad \; } &
      \quad \;  SG: \hspace{2ex} ${\bf r}_1 =(-1,\overline{0})$ \\        &
      \hspace{8.5ex} $ {\bf r}_2 =(0,1,\overline{0},1)$             \\ \hline
      \multirow{2}{*}{ \quad \; $F({\bf x})=\psi(z_1,z_2)={\bf x}^{{\bf r}_1}\left(p_1+p_2{\bf x}^{{\bf r}_2}\right)$ \quad \;}  &
      \multirow{2}{*}{ \quad \;  AKP:  \hspace{1ex} formula \eqref{eq1} } \\ & \\ \hline
      \multirow{2}{*}{ \quad \; $F({\bf x})=\psi(z_1,z_2,z_3)={\bf x}^{{\bf r}_1}\left(p_1+p_2{\bf x}^{{\bf r}_2}+p_3{\bf x}^{{\bf r}_3}\right)$ \quad \;}  &
      \multirow{2}{*}{\quad \;  BKP:  \hspace{1ex} formula \eqref{eq2} } \\ & \\ \hline
  \end{tabular}
\end{table}
\end{center}

We will consider two families of Poisson structures that are each defined by
a constant skew-symmetric matrix $T$. Constant Poisson structures defined by
\begin{equation}
\label{eq:c_ps}
\pb{x_i,x_j}=\Omega_{i,j}({\bf x})=T_{i,j}, \quad i,j\in\set{0,1,\ldots,n-1}
\end{equation}
and quadratic Poisson structures which are known as log-canonical
(or diagonal, or Lotka-Volterra) Poisson structures \cite{dufour}. These are defined by the brackets
\begin{equation}
\label{eq:lv_ps}
\pb{x_i,x_j}=\Omega_{i,j}({\bf x})=T_{i,j}x_ix_j, \quad i,j\in\set{0,1,\ldots,n-1}.
\end{equation}
It is well-known that, because of the skew-symmetry of $T$, such brackets always satisfy the Jacobi identity, hence they
are indeed Poisson brackets \cite{polpoisson}. The rank of these Poisson structures,
at a generic point, equals the rank of the constant matrix $T$ and their Casimirs are in
correspondence with the null-vectors of $T$. If ${\bf k}=(k_0,k_1,\ldots,k_{n-1})$ is a
null-vector of $T$ then ${\bf k \cdot x}$ is a Casimir of the constant Poisson structure
while $\bf x^k$ is a Casimir of the quadratic Poisson structure. The log-canonical term
comes from the fact that these quadratic structures are related to the constant Poisson
structures by exponentiation of the coordinates. If we define new variables $\eta_i=e^{x_i}$,
where the $x_i$ variables satisfy \eqref{eq:c_ps}, then the bracket of $\eta_i$ and $\eta_j$ is
$$
\pb{\eta_i,\eta_j}=T_{i,j}\eta_i\eta_j.
$$
Because of this relation and because the examples we give in Section \ref{sec:ex}
are of the form \eqref{eq:second_form}, in what follows we will focus on difference equations of the
form \eqref{eq:second_form} that preserve quadratic Poisson structures of the form
\eqref{eq:lv_ps}. The explicit relation between maps of the form \eqref{eq:first_form}
and \eqref{eq:second_form} is given in the following lemma.
\begin{lemma}
\label{lem:rel_of_two_forms}
Suppose that $M:\mathbb{R}^n\rightarrow\mathbb{R}^n, {\bf x} \mapsto {\bf x}'$ is a map of
the form \eqref{eq:map_gener}-\eqref{eq:def_of_M} where $F({\bf x})=\phi(y_1,y_2,\ldots,y_k)$
as in \eqref{eq:first_form}. Then $M$ preserves
the constant Poisson structure $\pb{x_i,x_j}=T_{i,j}$ if and only if the map
$L:\mathbb{R}^n_{>0}\rightarrow\mathbb{R}^n_{>0}$ defined by
${\bf v}=(v_0,v_1,\ldots,v_{n-1}) \mapsto {\bf v}'=(v_1,v_2,\ldots,v_{n-1},G({\bf v}))$
with $G(v_0,v_1,\ldots,v_{n-1})=e^{F({\bf x})}$ and $v_i=e^{x_i}$
preserves the quadratic Poisson structure $\pb{v_i,v_j}=T_{i,j}v_iv_j$.
The map $L$ is of the form
\eqref{eq:map_gener}-\eqref{eq:def_of_M} with $G$ as in \eqref{eq:second_form}.
\end{lemma}
\begin{proof}
By definition, $G$ is indeed of the form \eqref{eq:second_form} and is
defined by the function $\psi(z_1,z_2,\ldots,z_k)=e^{\phi(\ln z_1,\ln z_2,\ldots,\ln z_k)}$.
To verify that $L$ preserves the Poisson structure $\pb{v_i,v_j}=T_{i,j}v_iv_j$ we only need
to verify \eqref{eq:pois_pres} for $f=v_i, i=0,1,\ldots,n-2$ and $g=v_{n-1}$. We have,
for any $i=0,1,\ldots,n-2$,
\begin{align*}
\pb{v_{i},v_{n-1}}'=&
T_{i,n-1}v_{i+1}G({\bf v})=\sum_{j=0}^{n-1}T_{i+1,j}v_{i+1}G({\bf v})
\frac{\partial F({\bf x})}{\partial x_j}\\
=&\sum_{j=0}^{n-1}T_{i+1,j}v_{i+1}v_jG({\bf v})
\frac{\partial F({\bf x})}{\partial v_j}
%=\sum_{j=0}^{n-1}T_{i+1,j}v_{i+1}v_j\frac{\partial G({\bf v})}{\partial v_j}
\\
=&\sum_{j=0}^{n-1}T_{i+1,j}v_{i+1}v_j\frac{\partial G({\bf v})}{\partial v_j}=
\pb{v_i',v_{n-1}'},
\end{align*}
where in the second equality we have used our assumption that the map $M$
preserves the constant Poisson structure $(T_{i,j})$. The proof of the
other direction is done similarly.
\end{proof}
\begin{rem}
The previous lemma allows us to present our results only for maps of the
form \eqref{eq:second_form} and for quadratic Poisson structures. Then
the same results will hold true for maps of the form \eqref{eq:first_form}
and constant Poisson structures. One can prove a more general result to
cover a larger class of mappings and Poisson structures. Namely, under the
assumptions of the previous lemma, if $h:\mathbb R\rightarrow X\subseteq \mathbb R$
is any differentiable function with $h'(x)\neq 0$ for all $x\in\mathbb R$,
then the brackets $\pb{v_i,v_j}=T_{i,j}h'(h^{-1}(v_i))h'(h^{-1}(v_j))$ define
a Poisson structure on $X^n$ that is preserved by the map
$L:X^n\rightarrow X^n, {\bf v}=(v_0,v_1,\ldots,v_{n-1}) \mapsto
{\bf v}'=(v_1,v_2,\ldots,v_{n-1},G({\bf v}))$
with $G(v_0,v_1,\ldots,v_{n-1})=h(F({\bf x}))$ and $v_i=h(x_i)$.
\end{rem}

In Section \ref{sec:map_to_poisson} we show
that under some assumptions on the function $\psi$ we can always find
a quadratic Poisson structure of the form \eqref{eq:lv_ps} that is preserved by the
map \eqref{eq:map_gener}-\eqref{eq:def_of_M} [Theorem \ref{prop:pois_preserv}].
In Section \ref{sec:pois_to_map} we study the inverse
problem; given a (log-canonical) quadratic Poisson structure to find a map of
the form \eqref{eq:map_gener}-\eqref{eq:def_of_M} that preserves this structure
[Theorems \ref{prop:inv_prob_symm_even}, \ref{prop:inv_prob_symm_odd}].
Last, in Section \ref{sec:ex} we apply our theory to maps which are obtained
as reductions of known partial difference equations.

\section{Finding the Poisson structure given the difference equation}
\label{sec:map_to_poisson}

We begin this section by showing that if a Poisson structure is preserved
by a map of the form \eqref{eq:map_gener}-\eqref{eq:def_of_M}, then it must be of a specific
form and the function $F$ must satisfy certain PDE's that depend on the Poisson
structure. We show in the next lemma that if $F$ is of the form \eqref{eq:second_form} and
the Poisson structure is of the form \eqref{eq:lv_ps} with Toeplitz matrix $T$
then the PDE's are transformed into a linear system of equations.
Recall that a square $n\times n$ matrix $T$ is called Toeplitz if its entries
$T_{i,j}$ depend only on the differences $j-i$. This means that if $T$ is Toeplitz
there exist $2n-1$ numbers $T_{j}$ for $j=-n+1, -n+1, \ldots, n-1$ such that for any
$i,j\in\set{1,2, \ldots, n}, \; T_{i,j}=T_{j-i}$. Notice that in the case of a skew symmetric
Toeplitz matrix $T$, which is the case we consider, the numbers $T_j$ have the property
$T_{-j}=-T_j$, in particular $T_0=0$.
\begin{lemma}
\label{lem:gen_form}
Let $M$ be the map \eqref{eq:map_gener}-\eqref{eq:def_of_M}
and $\PB$ a Poisson structure with $\Omega({\bf x})$
the corresponding Poisson matrix.
\begin{enumerate}
\item
The map $M$ preserves the Poisson structure $\PB$ if and only if the following
two relations are satisfied
\begin{equation*}
%\label{poisson_relation1}
\Omega_{i+1,j+1}({\bf x})=\Omega_{i,j}({\bf x'}):=\Omega_{i,j}(M({\bf x})),
\quad \text{ for all } \quad 0\leq i<j<n-1
\end{equation*}
and
\begin{gather*}
\begin{split}
%\label{poisson_relation2}
\Omega_{i,n-1}({\bf x'})=\pb{x_{i+1},x_{n-1}'}:=\sum_{j=0}^{n-1}
\frac{\partial F}{\partial x_j}\Omega_{i+1,j}({\bf x}),
\quad \text{ for } \quad i=0,1,\ldots,n-2.
\end{split}
\end{gather*}
\item
If the function $F$ of \eqref{eq:second_form} is of the form $F({\bf x})=z_1\tilde{\psi}(z_2,\ldots,z_k)$
with $\tilde{\psi}$ any function in $C^1(\mathbb R^{k-1})$ and if $\PB$ is a quadratic
Poisson structure of the form \eqref{eq:lv_ps} with Toeplitz matrix
$T$, then $M$ preserves $\PB$ if and only if
\begin{gather}
\begin{split}
\label{eq:linear_system_2}
\sum_{j=0}^{n-1}r_{1,j}T_{j-i}=T_{n-i},& \quad \text{for}\quad i=1,\ldots, n-1\quad \text{and}\\
\sum_{j=0}^{n-1}r_{\ell,j}T_{j-i}=0,\quad \text{for}\quad  i=1,&\ldots, n-1, \quad
\ell=2,3,\ldots,k.
\end{split}
\end{gather}
\end{enumerate}
\end{lemma}
\begin{proof}
Item (1) is easily proved by direct computation using formula \eqref{eq:pois_pres}.
For the proof of item (2) notice that, because $T$ is Toeplitz the first system of equations
of item (1) is automatically satisfied while the second one becomes
\begin{gather}
\sum_{j=0}^{n-1}\frac{\partial F}{\partial x_j}x_{i+1}x_jT_{j-i-1}=
x_{i+1}x_nT_{n-i-1} \nonumber  \iff \\ \label{eq:system_function}
\sum_{\ell=1}^k\sum_{j=0}^{n-1}\frac{\partial F}{\partial z_\ell}
\frac{\partial z_\ell}{\partial x_j} x_jT_{j-i-1}=F({\bf x})T_{n-i-1},
\quad i=0,1,\ldots,n-2.
\end{gather}
Using that $F({\bf x})=z_1\tilde{\psi}(z_2,\ldots,z_k)$ and that
$x_j\frac{\partial z_\ell}{\partial x_j}=r_{\ell,j}z_{\ell}$,
system \eqref{eq:system_function} is transformed into the
second part \eqref{eq:linear_system_2}.
\end{proof}

The linear system \eqref{eq:linear_system_2} has $k\cdot (n-1)$ equations
and $n-1$ variables, therefore is unlikely to have a solution.
Imposing some restrictions on the vectors
${\bf r}_i$ we are able to reduce the size of the system and in some
cases obtain general (non)existence results.
We do that in the next lemma after introducing some notation.

\begin{notation}
\label{not:hankel_matrix}
If ${\bf r}$ is any row vector we write ${\bf r}^*$ for the vector
obtained from ${\bf r}$ by deleting its first element and we write
$\hat{\bf r}$ for the vector obtained by reversing the order of the
entries of $\bf r$. We say that the vector $\bf r$ is symmetric if
${\bf r}=\hat{\bf r}$ and that is skew-symmetric if ${\bf r}=-\hat{\bf r}$.
We will also write $T^*$ for the $(n-1)\times (n-1)$ minor of $T$
obtained by deleting its first row and column and $Q$ for the
$(n-1)\times n$ minor of $T$ obtained by deleting its first row.
The $n\times n$ Hankel matrix $J$, defined by $J_{i,n+1-i}=1$ for all
$i=1,2,\ldots, n$, and all other entries zero, will be useful.
\end{notation}
Using the above notation a Toeplitz $n\times n$ matrix
$T$ is skew-symmetric (symmetric) if and only if
$JTJ=-T\; (JTJ=T)$. For example,
it is easy to see that $J^2=I$
and for the skew-symmetric matrix $\Omega$
in \eqref{eq:sympl_matr}, $J\Omega J=-\Omega$.
Similarly, the vector ${\bf r}$ is skew-symmetric (symmetric) if and only if
$J{\bf r}=-{\bf r}$ ($J{\bf r}={\bf r}$).

With the above notation the linear system \eqref{eq:linear_system_2} is written, in an equivalent
matrix form, as
\begin{gather}
\label{eq:linear_system_2_inv_matrix}
\begin{align}
-r_{1,0}\,{\bf t}+&T^*\,{\bf r}_1^{*t}= \hat{{\bf t}},\qquad
-r_{\ell,0}\,{\bf t}+T^*{\bf r}_\ell^{*t}= 0,\quad \ell=2,3,\ldots,k,
\end{align}
\end{gather}
or equivalently again, as
\begin{gather}
\label{eq:linear_system_2_inv_matrix_new}
\begin{split}
Q\,{\bf r}_1^t=
\hat{\bf t}, \quad
Q\,{\bf r}_\ell^t=0,\quad \ell=2,3,\ldots,k\,,
\end{split}
\end{gather}

where ${\bf t}$ is the vector
$$
{\bf t}^t=
\begin{pmatrix}
T_1 &
T_2 &
\cdots &
T_{n-1}
\end{pmatrix}\,.
$$

\begin{lemma}\
\label{lem:simple_form_palindr}
\begin{enumerate}
\item If $r_{1,0}=-1$ (resp. $r_{1,0}=1$) and $ r_{\ell,0}=0$ for $\ell=2,3,\ldots, n$ and if the vectors
${\bf r}_\ell^*$ are symmetric (resp. skew-symmetric) for all $\ell=1,2,\ldots, k$,
then the system \eqref{eq:linear_system_2_inv_matrix} becomes half in size. More
explicitly, for any $\ell\in\set{1,2,\ldots, k}$, the vector $T^*{\bf r}_\ell^{*t}$ is skew-symmetric
(resp. symmetric).
\item
If $n$ is even, $k=2,\; r_{1,0}=-1,\; r_{2,0}=0$ and if the vector
${\bf r}_\ell^*$ is symmetric for $\ell=1,2$ then the linear system
\eqref{eq:linear_system_2_inv_matrix} has a non-trivial solution.
\item
If ${\bf r}_{\ell}=(r_\ell,r_\ell,\ldots,r_\ell)$ for some $r_\ell\in\mathbb R$
then the $\ell$-th equations
of \eqref{eq:linear_system_2_inv_matrix} simplify to
$$
T_i(r_\ell+1)+\sum_{j=i+1}^{n-i-1}r_\ell T_j-T_{n-i}=0,
\; i=1,2,\ldots,[\frac{n}{2}]\,.
$$
\end{enumerate}
\end{lemma}
\begin{proof}
For $r_{1,0}=-1$ the vector $\hat{\bf t}+r_{1,0}{\bf t}$ is skew-symmetric
while for $r_{1,0}=1$ is symmetric. In order to prove item (1) it is enough to show that the vector
$T^*{\bf r}_\ell^{*t}$ is skew-symmetric (resp. symmetric). Let
$J$ be the $(n-1)\times (n-1)$ Hankel matrix defined in
Notation \ref{not:hankel_matrix}. Then
$$
J\,T^*\,{\bf r}_\ell^{*t}=-J^2\,T^*\,J\,{\bf r}_\ell^{*t}=
-T^*{\bf r}_\ell^{*t},
$$
where we have used the skew-symmetry of $T^*$ and the symmetry of ${\bf r}_\ell^*$.
The symmetric case is done similarly. For item (2) notice that because
$n$ is even and because the $n-1$ dimensional vector $T^*{\bf r}_\ell^{*t}$
is skew-symmetric, its middle element is zero. Thus the (homogeneous)
linear system \eqref{eq:linear_system_2_inv_matrix} has $2(\frac{n}{2}-1)=n-2$
equations with $n-1$ variables ($T_1, T_2, \ldots, T_{n-1}$) and therefore a non-trivial solution.
Item (3) follows by direct computation.
\end{proof}%
The next theorem about maps of the form \eqref{eq:map_gener}-\eqref{eq:def_of_M}
is a corollary of Lemma \ref{lem:simple_form_palindr}.
\begin{thm}
\label{prop:pois_preserv}
Let $n$ be even and $M$ the map \eqref{eq:map_gener}-\eqref{eq:def_of_M}
with the function $F$ of \eqref{eq:second_form} being of the form
$$
F({\bf x})={\bf x}^{r_1} \tilde{\psi}({\bf x}^{r_2})
$$
for some real function $\tilde{\psi}\in C^1(\mathbb R)$. If
$r_{1,0}=-1, r_{2,0}=0$ and ${\bf r}_\ell^*$ symmetric for
$\ell=1,2$ then there is a quadratic Poisson structure
$\pb{x_i,x_j}=T_{i,j}x_ix_j$ that is preserved by the map $M$.
The matrix $T$ is a skew-symmetric Toeplitz matrix with first
row $(0 \; T_1 \; T_2 \; \ldots \; T_{n-1})$, where the $T_i$
are determined by the non-trivial solution of
\eqref{eq:linear_system_2_inv_matrix}.
\end{thm}
%
%Applying our results to the maps given in Table \ref{tbl:table1}
%we see that there exist Poisson structures that are preserved by
%the even dimensional SG and KdV maps (using Lemma
%\ref{lem:rel_of_two_forms} in the latter case).
In the rest of
this section we present Poisson structures that are preserved by the maps
of Table \ref{tbl:table1}. First, notice that the vectors
${\bf r}_i$ which define the mKdV and pKdV maps are identical.
Therefore, according to Lemma \ref{lem:rel_of_two_forms},
a constant Poisson structure $T$ is preserved by the
pKdV map if and only if the quadratic log-canonical Poisson
structure defined by the matrix $T$ is preserved by the
mKdV map. This is in accordance with the results of \cite{dihn_3}.
An easy calculation, using \eqref{eq:linear_system_2} or
\eqref{eq:linear_system_2_inv_matrix}, shows that in even dimensions the SG
map preserves the non-degenerate log-canonical Poisson structure
defined by the Toeplitz matrix with first line
$(\overline{0,1})$. The KdV map preserves a constant Poisson
structure with Toeplitz matrix $T$ where, for $n\equiv 2\mod 4$ the first line of $T$
is $(\overline{0,1,-1,0},0,1)$ and for $n\equiv 0\mod 4$ is $(\overline{0, 0, 1, -1})$.
These Poisson structures are non-degenerate. For
$n\equiv 3\mod 4$ the first line of $T$ is
$(\overline{0,1,0,-1},0,1,0)$ and $T$ is degenerate with
rank 2.
For all other remaining cases the linear system
\eqref{eq:linear_system_2} does not have a solution and therefore
there are no log-canonical (respectively constant, for the
odd-dimensional SG map) Poisson structures that are
preserved. In the next proposition we show that
reductions of these remaining cases give rise to maps that preserve
Poisson structures of our form (see also Proposition \ref{prop:redu}).

%For $n\equiv 3\mod 4$, the linear system
%\eqref{eq:linear_system_2} does not have a solution and therefore
%there is not a constant Poisson structure that is preserved.

\begin{prop}\ %
\label{prop:reductions}
\begin{enumerate}
\item For $n\equiv 1\mod 4$ the reduction $w_i=x_ix_{i+1}$ of the KdV map
gives rise to a map of the same form \eqref{eq:map_gener}-\eqref{eq:def_of_M}
with function $F$ as in \eqref{eq:second_form} that preserves the log-canonical
Poisson structure defined by the non-degenerate Toeplitz matrix with first line
$(0,\overline{0,1,0,-1},0,1,0)$.
\item For the odd dimensional SG map the reduction $w_i=x_ix_{i+1}$
gives rise to a map of the same form \eqref{eq:map_gener}-\eqref{eq:def_of_M}
with function $F$ as in \eqref{eq:second_form} that preserves the log-canonical
Poisson structure defined by the non-degenerate Toeplitz matrix with first line
$(0,\overline{1})$.
\item For the even (resp. odd) dimensional mKdV map the reduction
$w_i=\frac{x_{i+2}}{x_i}$ (resp. $w_i=\frac{x_{i+1}}{x_i}$) gives rise
to a map  of the form \eqref{eq:map_gener}-\eqref{eq:def_of_M}
with function $F$ as in \eqref{eq:second_form} that preserves the
log-canonical Poisson structure defined by the non-degenerate
Toeplitz matrix with first line $(0, 1,\overline{0})$
(resp. $(0, 1, \overline{-1, 1})$).
\end{enumerate}
\end{prop}
\begin{proof}
We only give the vectors $\bf{r}_i$ obtained after the
reductions since the rest are straightforward computations.
For the KdV reductions the vectors ${\bf r}_1$ and ${\bf r}_2$ are
${\bf r}_1=(\overline{-1,0})$ and
${\bf r}_2=(0,0,\overline{1},0)$
and for the SG reductions they are
${\bf r}_1=(\overline{-1,1})$ and
${\bf r}_2=(0,1,\overline{-1,1})$.
For the even dimensional mKdV map they are given by
${\bf r}_1=(\overline{-1,0})$ and
${\bf r}_2=(\overline{0,1})$ and for the odd
dimensional mKdV map by
${\bf r}_1=(\overline{-1})$ and
${\bf r}_2=(0,\overline{1})$.
%Note that the order of the reduced map is one less
%of the order of the original map for the cases of
%items (1) and (3) while is two less for
%the even dimensional $mKdV$.
\end{proof}

\section{Finding the difference equation given the Poisson structure}
\label{sec:pois_to_map}

We consider now the inverse problem of finding a difference equation
of the form \eqref{eq:diff_equat} that preserves a given Poisson structure.
In what follows we assume that the matrix $T$ is skew-symmetric
and Toeplitz. We first show that if the map $M$ defined by
\eqref{eq:map_gener}-\eqref{eq:def_of_M}
preserves the quadratic log-canonical Poisson structure with matrix $T$
then the function $F$ is necessarily of the form \eqref{eq:second_form}.

\begin{prop}
\label{prop:form_of_M}
Let $M$ be the map \eqref{eq:map_gener}-\eqref{eq:def_of_M}
which preserves a quadratic log-canonical Poisson structure with matrix $T$.
Then the function $F$ defining $M$ is of the form \eqref{eq:second_form}.
\end{prop}
\begin{proof}
According to Lemma \ref{lem:gen_form}, if the map $M$ preserves the quadratic Poisson structure with
matrix $T$ then
\begin{equation}
\label{poisson_relation2}
\sum_{j=0}^{n-1}\frac{\partial F}{\partial x_j}T_{j-i}x_j=T_{n-i}F({\bf x}),
\quad i=1,2,\ldots,n-1\,.
\end{equation}
If $\bf r_1$ is a solution of the non-homogeneous linear system \eqref{eq:linear_system_2} then
we can verify that $F_1={\bf x}^{\bf r_1}$ is a solution of \eqref{poisson_relation2}. Since
the ratio of two solutions of \eqref{poisson_relation2} is a solution of the corresponding
homogeneous system it follows that its general solution
is $F={\bf x}^{\bf r_1}\tilde{F}(\bf x)$ where $\tilde{F}$ is the solution
of the homogeneous one. The system
\begin{equation}
\label{eq:hom_pde}
\sum_{j=0}^{n-1}\frac{\partial F}{\partial x_j}T_{j-i}x_j=0,
\quad i=1,2,\ldots,n-1\,,
\end{equation}
is linear and can be solved using the method of characteristics (see \cite{pdes}). 
It can be verified directly that if ${\bf r}_\ell$ is a solution of the
homogeneous part of \eqref{eq:linear_system_2} then a solution
$\tilde{F}$ of \eqref{eq:hom_pde} will remain constant along the
surface defined by ${\bf x}^{\bf r_\ell}=C$. This shows that
the solution of \eqref{poisson_relation2} is
$F({\bf x})={\bf x}^{{\bf r}_1}\tilde{\psi}({\bf x}^{{\bf r}_2},{\bf x}^{{\bf r}_3},\ldots,{\bf x}^{{\bf r}_k})$
where $\tilde{\psi}$ is any real function of $k-1$ variables and
${\bf r}_\ell,$ for $\ell=2,3,\ldots, k$, are solutions of the homogeneous part
of \eqref{eq:linear_system_2}. Therefore $F$ is indeed of the form
\eqref{eq:second_form}.
\end{proof}

The proof of the previous proposition shows that the existence of a map of the form
\eqref{eq:map_gener}-\eqref{eq:def_of_M} preserving a given log-canonical Poisson
structure amounts to a solution of a linear system. Assuming that the matrix
$T$ has sufficiently large rank then we can derive existence and non-existence results
about maps that preserve the corresponding Poisson structure.
%This is done in the next proposition.
%
\begin{prop}
\label{prop:pois_preserv_inv}
Let $Q$ be the matrix obtained from $T$ by deleting its first row, as in Notation \ref{not:hankel_matrix},
and $\PB$ the quadratic Poisson structure of the form \eqref{eq:lv_ps} with matrix $T$.
\begin{enumerate}
\item If $Q$ is of maximal rank, then there exists a map $M$
of the form \eqref{eq:map_gener}-\eqref{eq:def_of_M} with function $F$
of the form \eqref{eq:second_form} which preserves
the Poisson structure $\PB$.
\item If $Q$ is of rank $m\leq n-1$ and $M$ is a map of the form
\eqref{eq:map_gener}-\eqref{eq:def_of_M} which preserves
the Poisson structure $\PB$, then $F$ is of the form
\eqref{eq:second_form} and the vectors ${\bf r}_\ell,$ for
$\ell=2,3,\ldots,k$, form a linear space of dimension smaller or
equal than $n-m$.
\item Assuming that the vectors ${\bf r}_\ell$ for $\ell=2,3,\ldots,k$ are
linearly independent, $Q$ is of rank $m\leq n-1$ and $M$ the map
\eqref{eq:map_gener}-\eqref{eq:def_of_M} with $F$ of the form
\eqref{eq:second_form} with $k>n-m+1$, then $M$ does not preserve
the Poisson structure $\PB$.
\end{enumerate}
\end{prop}
\begin{proof}
For item (1) we first note that from Proposition \ref{prop:form_of_M},
if the map M preserves the Poisson structure $\PB$ then the function $F$
defining $M$ is necessarily of the form \eqref{eq:second_form}. The
existence of such $M$ is guaranteed by the existence of solutions of
the linear system \eqref{eq:linear_system_2_inv_matrix_new}
which consist of $k$ linear systems (in the ${\bf r}_\ell$'s) each one
having $n-1$ equations and $n$ variables. Because $Q$ is of maximal rank
its non-homogeneous part has a solution, and due to their dimension, the
rest homogeneous $k-1$ linear systems have a non-trivial solution.
The proof of items (2) and (3) is a consequence of the dimension of the homogeneous
part of the linear system \eqref{eq:linear_system_2_inv_matrix_new}.
\end{proof}

We now show that for the non-degenerate Poisson structures, the maps that preserve them
given in item (1) of the previous proposition, have $r_{1,0}=-1$,
$r_{\ell,0}=0$ for $\ell=2,\ldots, k$ and the vectors ${\bf r}^*_\ell$
are symmetric for $\ell=1,2,\ldots, k$ (cf. \cite{toep2,toep1}).
This serves as a partial inverse of Theorem \ref{prop:pois_preserv}.
\begin{thm}
\label{prop:inv_prob_symm_even}
Let $n$ be even and $T$ an $n\times n$ matrix of full rank. Also let $\PB$ be the
Poisson structure $\pb{x_i,x_j}=T_{i,j}x_ix_j$ and $M$ a map of the
form \eqref{eq:map_gener}-\eqref{eq:def_of_M} which preserves $\PB$.
Then the function $F$ defining $M$ is of the form \eqref{eq:second_form}
with
$$
F({\bf x})=\frac{({\bf x}^*)^{{\bf r}_1^*}}{x_0}\,\tilde{\psi}(({\bf x}^*)^{{\bf r}_2^*})\,,
$$
where the vectors ${\bf r}^*_1, {\bf r}^*_2$ are symmetric. They are explicitly
given by the formulas
\begin{equation*}
r_{1,j}=\frac{\det(T^{(j+1,1)})}{\det(T)}, \quad j=1,\ldots, n-1.
\end{equation*}
The matrix $T^{(j,1)}$ is obtained from $T$ by replacing its
$j$-th column by the vector
$\begin{pmatrix}
a&
T_{n-1}&
T_{n-2}&
\cdots&
T_1
\end{pmatrix}^t$
where $a\in\mathbb R$ and
$$
r_{2,j}=\operatorname{cofactor}(T,1,j+1), \quad j=0,1,\ldots, n-1.
$$
The $\operatorname{cofactor}(T,1,j+1)$ is the signed determinant
of the minor of $T$ obtained by deleting its first row and $j+1$
column.
\end{thm}
\begin{proof}
From the previous proposition it follows that $k\leq 2$ and, because
of the rank of $T$, it is sufficient to show that
the linear systems (in ${\bf r}_1^{*t}, {\bf r}_2^{*t}$)
\begin{gather*}
\begin{split}
T^* \, {\bf r}_1^{*t}
=
\begin{pmatrix}
T_{n-1}-T_1 \\
T_{n-2}-T_2 \\
\vdots \\
T_1-T_{n-1}
\end{pmatrix},
\quad
T^* \, {\bf r}_2^{*t}=0
\end{split}
\end{gather*}
have symmetric solutions. That, will be a consequence of the following
more general result: For an $m\times m$ skew-symmetric Toeplitz
matrix $R$ of rank $m-1$ (hence $m$ is odd) and ${\bf b}\in\mathbb R^m$
skew-symmetric, the solutions of the linear system $R\,{\bf q}={\bf b}$
are symmetric. We recall from Notation
\ref{not:hankel_matrix}, that $J$ is the $m\times m$ matrix with entries
$J_{i,m+1-i}=1$ for all $i=1,2,\ldots, m$ and all other entries zero.
Then $J\,R\,J=-R$ and $J\,{\bf b}=-{\bf b}$. This shows that
$R\,J\,{\bf q}=-J\,R\,J^2\,{\bf q}=-J\,R\,{\bf q}={\bf b}$ and therefore
the vector ${\bf q} - J\, {\bf q}$ is a (skew-symmetric) null vector of $R$.
Showing that $R$ has a non-zero symmetric null vector it will imply
(because of the rank of $R$) that ${\bf q} - J\, {\bf q}=0$ and therefore
${\bf q}$ is symmetric.

Let ${\bf v}$ be a null vector of $R$. The previous proof
(with ${\bf q}={\bf v}$ and ${\bf b}=0$) shows that $J\,{\bf v}$ is also
a null vector of $R$ and therefore ${\bf v}+J\,{\bf v}$
is a null vector of $R$ as well. So, we can assume that ${\bf v}$
is symmetric and the proof of Lemma \ref{lem:simple_form_palindr}
(item 2), shows that the homogeneous linear system $R\,{\bf v}=0$ has
$\frac{m-1}{2}$ equations with $\frac{m-1}{2}+1$ unknowns,
therefore a non trivial solution.

Having established that if $T$ is an $n\times n$ skew-symmetric Toeplitz
matrix of full rank the solution of the linear system
\eqref{eq:linear_system_2_inv_matrix_new} has $r_{1,0}=-1$ and
$r_{2,0}=0$ we can now use Cramer's rule to give explicit
formulas for the ${\bf r}_\ell$ for $\ell=1,2$. The linear system
\eqref{eq:linear_system_2_inv_matrix_new} is equivalently written as
\begin{gather}
\label{eq:matricial_sys_using_T}
T\,{\bf r}_1^t=
\begin{pmatrix}
a\\
T_{n-1} \\
T_{n-2} \\
\vdots \\
T_1
\end{pmatrix}, \quad
T\,{\bf r}_2^t=
\begin{pmatrix}
b\\
0 \\
0 \\
\vdots \\
0
\end{pmatrix}
\end{gather}
for arbitrary $a, b\in\mathbb R$. The one degree of freedom of
the linear system \eqref{eq:linear_system_2_inv_matrix_new} is
imposed into the parameters $a, b$. Cramer's rule gives that
\begin{equation}
r_{1,j}=\frac{\det(T^{(j+1,1)})}{\det(T)}, \quad, j=0,1,\ldots, n-1
\end{equation}
and similarly,
\begin{equation*}
r_{2,j}=\frac{\det(T^{(j+1,2)})}{\det(T)},
\quad j=0,1,\ldots, n-1.
\end{equation*}
Expanding the determinant $\det(T^{(j+1,2)})$ with
respect to its $j+1$ column we get
$$
r_{2,j}=\frac{b\operatorname{cofactor}(T,1,j+1)}{\det(T)},
\quad j=0,1,\ldots, n-1.
$$
The factor $\frac{b}{\det(T)}$ of the vector ${\bf r}_2$
can be absorbed into the arbitrary function $\tilde{\psi}$ and we get 
\begin{equation}
r_{2,j}=\operatorname{cofactor}(T,1,j+1), \quad j=0,1,\ldots, n-1.
\end{equation}
Also, because of the dimension of $T^*$, $r_{2,0}=\operatorname{cofactor}(T,1,1)=\det(T^*)=0$.
\end{proof}

\begin{rem}
\label{rem:arbitrary_last_elem}
Because $r_{2,0}=0$ and $T^* \, {\bf r}_2^{*t}=0$, the vector
${\bf r}_2$ does not depend on the entry $T_{n-1}$ of $T$.
This is consistent with the next example and with the results of
Table \ref{tbl:table2}. Also from the explicit form of the map $M$
given in the previous proof it follows that the map $M$ is invertible
(it can be solved for $x_0$), and also reversible, i.e.
$L^{-1}ML=M^{-1}$ for a suitable map $L$. In our case the map
$L$ is the involution ${\bf x}\mapsto \hat{\bf x}$.
\end{rem}

\begin{example}
\label{ex:1}
Suppose $n$ is even and $T_i=T_{i+1}$ for all $i<n-1$, i.e.
the first line of the matrix $T$ is
$$
(0,\overline{T_1},T_{n-1})=(0, T_1, T_1, \ldots, T_1, T_{n-1}).
$$
For generic values of $T_1, T_{n-1}$ the matrix $T$ is non-degenerate
with determinant $\det(T)=T_{n-1}^2T_1^{n-2}$ and the solution of
\eqref{eq:linear_system_2_inv_matrix_new} for ${\bf r}_2$ is
${\bf r}_2=(0,\overline{1,-1},1).$
Similarly for $T_i=-T_{i+1}$ for all $i<n-1$, the first line of $T$ becomes
$$
(0, \overline{T_1, -T_1}, T_{n-1})=(0, T_1, -T_1, \ldots, T_1, -T_1, T_{n-1})
$$
and the solution of \eqref{eq:linear_system_2_inv_matrix_new}
is ${\bf r}_2=(0,\overline{1})$.
\end{example}

We now consider the odd dimensional case.
\begin{thm}
\label{prop:inv_prob_symm_odd}
Let $n$ be odd and $T$ an $n\times n$ matrix with  $T^*$ of rank $n-1$.
We assume that $F$ is a function of the form \eqref{eq:second_form}
that defines the map \eqref{eq:map_gener}-\eqref{eq:def_of_M}
which preserves the Poisson structure $\pb{x_i,x_j}=T_{i,j}x_ix_j$.
Then if $k=2$ the vector ${\bf r}_2$ is symmetric and
can be chosen such that $r_{2,0}=1, \, {\bf r}_1^*=-\hat{\bf r}_2^*$
and $r_{1,0}=0$. Therefore the function $F$ is of the form
$$
F({\bf x})=({\bf x}^*)^{-\hat{\bf r}_{2}^*}\tilde{\psi}({\bf x}^{{\bf r}_{2}}).
$$
\end{thm}
\begin{proof}
According to Proposition \ref{prop:form_of_M} the function $F$ defining $M$
is of the form $F({\bf x})={\bf x}^{{\bf r}_1}\tilde{\psi}({\bf x}^{{\bf r}_2})$.
From the proof of the previous theorem, the vector ${\bf r}_2$
being a null vector of $T$, is symmetric.
The defining relations of the vectors ${\bf r}_1$ and ${\bf r}_2$ are
the linear systems \eqref{eq:linear_system_2_inv_matrix}
which, from our assumption that $T^*$ has full rank,
they have a unique solution in ${\bf r}_1^*, {\bf r}_2^*$. The
arbitrary function $\tilde{\psi}$ absorbs the parameters $r_{1,0}$ and $r_{2,0}$
and can be chosen to be equal to $0$ and $1$ respectively. The skew-symmetry and
Toeplitz form of $T$ gives that ${\bf r}_1^*=-\hat{\bf r}_2^*$.
\end{proof}

\begin{rem}
If the vector ${\bf t}=(T_1, T_2, \ldots, T_{n-1})$ is symmetric then
the vector ${\bf r}_1^*$ can be absorbed into the arbitrary function
$\psi$ and we can choose ${\bf r}_1=(-1,\bar{0})$ while if
$\bf t$ is skew-symmetric we can choose ${\bf r}_1=(1,\bar{0})$.
\end{rem}
\begin{rem}
The explicit form of the map given in the previous proposition shows that
the map is invertible if and only if the function $\tilde{\psi}$ is invertible and it
is reversible, with reversing symmetry the same map $L$ as in the even dimensional case,
if and only if $\tilde{\psi}^{-1}=\tilde{\psi}$.
\end{rem}
\begin{example}
Continuing Example \ref{ex:1} for odd $n$, if $T_i=T_{i+1}\neq 0$
for all $i<n-1$, i.e. if the first line of $T$ is
$(0, T_1, T_1, \ldots, T_1, T_{n-1})$ we get
$$
{\bf r}_2=(1,-\frac{T_{n-1}}{T_1},\frac{T_{n-1}}{T_1},
-\frac{T_{n-1}}{T_1},\ldots,-\frac{T_{n-1}}{T_1},1).
$$
Similarly for $T_i=-T_{i+1}\neq 0$ for all $i<n-1$, we get
$$
{\bf r}_2=(1,-\frac{T_{n-1}}{T_1},-\frac{T_{n-1}}{T_1},
-\frac{T_{n-1}}{T_1},\ldots,-\frac{T_{n-1}}{T_1},1).
$$
\end{example}

\renewcommand{\arraystretch}{2}
\begin{center}
\begin{table}[ht]
  \caption{Non-degenerate Poisson structures and the function $F$ which defines the map 
  that preserves the Poisson structure.
%  The function $F$ is defined by the  arbitrary function $\tilde{\psi}$.
}
  \label{tbl:table2}
  \begin{tabular}{ | c | c | c | }
    \hline
     $(T_1,T_2,\ldots,T_{n-1})$ & Determinant & Form of the function $F$\\
     \hhline{|=|=|=|}
	  $(\overline{1}, t)$ &
      $t^2$ &
      $F=\frac{(x_1x_{n-1})^{t-1}}{x_0}\tilde{\psi}(x_{2n-1}\prod_{j=1}^{\frac{n}{2}-1}\frac{x_{2j-1}}{x_{2j}})$
      \\ \hline
      $(\overline{1,-1}, t)$ &
      $t^2$ &
      $F=\frac{(x_1x_{n-1})^{1-t}}{x_0}\tilde{\psi}(\prod_{j=1}^{n-1}x_{j})$
      \\ \hline
      $(\overline{1,0}, t)$ &
      $2^{n-4}(t+1)^2$ &
      $F=\frac{(x_2x_{n-2})^{2t-2}}{x_0}\tilde{\psi}(x_1x_{n-1})$
      \\ \hline
      $(1,\overline{0},t)$ &
      $(t+1)^2$ &
      $F=\frac{\prod_{j=1}^{\frac{n}{2}-1}x_{2j}^{t-1}}{x_0}\tilde{\psi}(\prod_{j=1}^{\frac{n}{2}}x_{2j-1})$
      \\ \hline
      $(0,1,\overline{0},1,t)$, $\frac{n}{2}$ odd &
      $t^2$&
      $F=\frac{\prod_{j=0}^{\frac{n-2}{4}}x_{4j+1}^t}{x_0}\tilde{\psi}(x_1\prod_{j=1}^{\frac{n-2}{4}}\frac{x_{4j+1}}{x_{4j-1}})$
      \\ \hline
      $(0,1,\overline{0},1,t)$, $\frac{n}{2}$ even &
      $16$&
      $F=\frac{\prod_{j=1}^{\frac{n}{2}}x_{2j-1}^\frac{t}{2}}{x_0}\tilde{\psi}(\prod_{j=1}^{\frac{n}{4}}x_{4j-2})$
      \\ \hline
      \end{tabular}
\end{table}
\end{center}

In Tables \ref{tbl:table2}, \ref{tbl:table3} and \ref{tbl:table4}
we give the vector $\bf t$ and the form of the
function $F$ which defines the map $M$ that preserves
the quadratic Poisson structure $\pb{x_i,x_j}=T_{i,j}x_ix_j$
where ${\bf t}=(T_1,\ldots,T_{n-1})$. In Table
\ref{tbl:table2} we present non-degenerate Poisson structures
which depend on a parameter $t\in\mathbb R$
and in Table \ref{tbl:table3} we present the same structures with
$t$ chosen such that the matrix $T$ is degenerate. In Table
\ref{tbl:table4} we present Poisson structures of odd
dimension $n$. The results of Table \ref{tbl:table2} verify
Remark \ref{rem:arbitrary_last_elem}, that
the vector ${\bf r}_2$ is not affected from the last entry of
the matrix $T$ which is taken arbitrary so that the matrix
is non-degenerate. In Table \ref{tbl:table3} the rank of the
Poisson structures is $n-2$ and therefore the function $F$
can be a two variable function. These examples illustrate
the results of Theorems \ref{prop:inv_prob_symm_even} and
\ref{prop:inv_prob_symm_odd}.

\renewcommand{\arraystretch}{2}
\begin{center}
\begin{table}[ht]
  \caption{Degenerate Poisson structures of even dimension and the function
  $F$ which defines the map that preserves the Poisson structure.
%  The function $F$ is defined by the arbitrary function $\tilde{\psi}$.
  }
  \label{tbl:table3}
  \begin{tabular}{ | c | c | c | }
    \hline
     $(T_1,T_2,\ldots,T_{n-1})$ & Form of the function $F$\\
     \hhline{|=|=|=|}
	  $(\overline{1}, 0)$ &
      $F=\frac{1}{x_1}\tilde{\psi}(x_0x_{n-1},x_0\prod_{j=1}^{\frac{n}{2}}\frac{x_{2j}}{x_{2j-1}})$
      \\ \hline
      $(\overline{1,-1}, 0)$ &
      $F=\frac{1}{x_1}\tilde{\psi}(\frac{x_{n-1}}{x_0},\prod_{j=1}^{n-1}x_{j-1})$
      \\ \hline
      $(\overline{1,0}, -1)$ &
      $F=\frac{1}{x_2}\tilde{\psi}(x_1x_{n-1},x_0x_{n-2})$
      \\ \hline
      $(1,\overline{0},-1)$ &
      $F=x_0\tilde{\psi}(\prod_{j=1}^{\frac{n}{2}}x_{2j-2},\prod_{j=1}^{\frac{n}{2}}x_{2j-1})$
      \\ \hline
      $(0,1,\overline{0},1,0)$, $\frac{n}{2}$ odd &
      $F=\frac{1}{x_0}\tilde{\psi}(x_{n-1}\prod_{j=1}^{\frac{n-2}{4}}\frac{x_{4j-3}}{x_{4j-1}},x_0\prod_{j=1}^{\frac{n-2}{4}}\frac{x_{4j}}{x_{4j+2}})$
      \\ \hline
      \end{tabular}
\end{table}
\end{center}

\renewcommand{\arraystretch}{2}
\begin{center}
\begin{table}[ht]
  \caption{Poisson structures of odd dimension and the function $F$
  which defines the map that preserves the Poisson structure.
%  The function $F$ is defined by the arbitrary function $\tilde{\psi}$.
  }
  \label{tbl:table4}
  \begin{tabular}{ | c | c | c | }
    \hline
     $(T_1,T_2,\ldots,T_{n-1})$ & Rank & Form of the function $F$\\
     \hhline{|=|=|=|}
	  $(\overline{1})$ &
      $n-1$ &
      $F=\frac{1}{x_0}\tilde{\psi}(x_0\prod_{j=1}^{\frac{n-1}{2}}\frac{x_{2j}}{x_{2j-1}})$
      \\ \hline
      $(\overline{1,-1})$ &
      $n-1$ &
      $F=x_0\tilde{\psi}(\prod_{j=0}^{n-1}x_{j})$
      \\ \hline
      $(\overline{1,0})$ &
      $n-1$ &
      $F=\frac{1}{x_1}\tilde{\psi}(x_0x_{n-1})$
      \\ \hline
      $(1,\overline{0})$ &
      $n-1$ &
      $F=\prod_{j=1}^{\frac{n-1}{2}}\frac{1}{x_{2j-1}}\tilde{\psi}(\prod_{j=0}^{\frac{n-1}{2}}x_{2j})$
      \\ \hline
%      $(0,0,1,0,\ldots0,1,t)$, $\frac{n}{2}$ odd &
%      $n-1$ &
%      $F=\frac{\prod_{j=0}^{\frac{n-2}{4}}x_{4j+1}^t}{x_0}\psi(x_1\prod_{j=1}^{\frac{n-2}{4}}\frac{x_{4j+1}}{x_{4j-1}})$
%      \\ \hline
%      $(0,0,1,0,\ldots0,1,t)$, $\frac{n}{2}$ even &
%      $n-1$ &
%      $F=\frac{\prod_{j=1}^{\frac{n}{2}}x_{2j-1}^\frac{t}{2}}{x_0}\psi(\prod_{j=1}^{\frac{n}{4}}x_{4j-2})$
%      \\ \hline
      \end{tabular}
\end{table}
\end{center}

\section{Poisson structures for known maps}
\label{sec:ex}

We now apply our results to several families of maps
and we find Poisson structures that they preserve. We also
find maps that preserve Poisson structures of specific form.
For simplicity we write $\LVPS({\bf t})$ (Lotka-Volterra
Poisson structure) for the quadratic Poisson structure
$\pb{x_i,x_j}=T_{i,j}x_ix_j$ where $T$ is a skew-symmetric
Toeplitz matrix and ${\bf t}=(T_1,\ldots,T_{n-1})$
with $T_{j-i}=T_{i,j}$ for all $0\leq i<j \leq n-1$.

First we consider maps which arise as reductions of the
AKP partial difference equation \cite{hirota,miwa}. These maps are defined by
an equation of the form
\begin{equation}
\label{eq:AKP_red}
A{\bf x}^{{\bf u}_0}x_{n}+B{\bf x}^{{\bf u}_1}+C{\bf x}^{{\bf u}_2}=0
\end{equation}
which is obtained from a $(z_1,z_2,z_3)$-travelling wave reduction of the AKP equation
%\begin{equation}
%\label{eq:AKP_part_diff}
%AT_{k+1,l,m}\tau_{k,1,1}+B\tau_{0,1,0}\tau_{1,0,1}+C\tau_{0,0,1}\tau_{1,1,0}=0.
%\end{equation}
\begin{equation}
\label{eq:AKP_part_diff}
A\mathcal{T}_{k+1,l,m}\mathcal{T}_{k,l+1,m+1}+B\mathcal{T}_{k,l+1,m}\mathcal{T}_{k+1,l,m}+C\mathcal{T}_{k,l,m+1}\mathcal{T}_{k+1,l+1,m}=0.
\end{equation}
We consider the $(z_1,z_2,z_3)$-travelling wave reduction
$\mathcal{T}_{k,l,m} = \tau_{z_1k+z_2l+z_3m}$ where $z_1, z_2, z_3\in\mathbb N$.
Because of the symmetry
of equation \eqref{eq:AKP_part_diff}, the order of $(z_1, z_2, z_3)$
is irrelevant and therefore we may use the constraint $0<z_1< z_2< z_3$.
By writing $\tau_n$ for $\tau_{z_1k+z_2l+z_3m+n}$
then the pullback of \eqref{eq:AKP_part_diff} under the transformation
$x_0=\frac{\tau_{-1}\tau_{1}}{\tau_0^2}$ is the map \eqref{eq:AKP_red} of
order $n=M(z_1,z_2,z_3)=z_2+z_3-z_1-2$. The vectors ${\bf u}^*_\ell$
in \eqref{eq:AKP_red} are symmetric and of dimension $n-1$.
For $z_3\geq z_1+z_2$ they are given by
\begin{gather*}
{\bf u}^*_0=(2,3,\ldots,z_2-1,\overline{z_2},z_2-1,\ldots,3,2)\\
{\bf u}^*_1=(\overline{0},1,2,\ldots,z_1-1,\overline{z_1},z_1-1, \ldots, 1, \overline{0}),\\
%{\bf u}^*_1=(\underbrace{0,0,\ldots,0}_{z_2-z_1-1
%\text{ terms}},1,2,\ldots,z_1-1,z_1,z_1,\ldots, z_1,\ldots),\\
{\bf u}^*_2=(0,\ldots,0)\,,
\end{gather*}
where the total number of zeros in ${\bf u}^*_1$
is $2(z_2-z_1-1)$.
If $z_3<z_1+z_2$ the exponents ${\bf u}^*_\ell$ coincide with
the exponents of the $(z_3-z_2, z_3-z_1, z_3)$ reduction.
Their first elements are respectively $u_{0,0}=1$ and
$u_{\ell,0}=0$ for $\ell=1,2$.

Equation \eqref{eq:AKP_part_diff} is a special case of a
more general partial difference equation, the BKP equation
\cite{miwa}
\begin{gather}
\label{eq:BKP_discr}
\begin{split}
A\mathcal{T}_{k+1,l,m}\mathcal{T}_{k,l+1,m+1}+&B\mathcal{T}_{k,l+1,m}\mathcal{T}_{k+1,l,m+1}+C\mathcal{T}_{k,l,m+1}\mathcal{T}_{k+1,l+1,m}+\\
&D\mathcal{T}_{k,l,m}\mathcal{T}_{k+1,l+1,m+1}=0.
\end{split}
\end{gather}
The same $(z_1,z_2,z_3)$-travelling wave reduction as before gives rise to
the $n$-th order map with $n=N(z_1,z_2,z_3)=z_1+z_2+z_3-2$, given by
\begin{equation}
\label{eq:BKP_red}
D{\bf x}^{{\bf u}_0}x_n+A{\bf x}^{{\bf u}_1}+B{\bf x}^{{\bf u}_2}+
C{\bf x}^{{\bf u}_3}=0.
\end{equation}
For $z_3\geq z_1+z_2$ the vectors ${\bf u}^*_\ell$ are
\begin{gather*}
{\bf u}^*_0=(2,3,\ldots,z_1+z_2-1,\overline{z_1+z_2},z_1+z_2-1,\ldots,3,2),\\
{\bf u}^*_1=(\overline{0},1,2,\ldots, z_2-1,\overline{z_2},z_2-1,\ldots, 2, 1, \overline{0}),\\
{\bf u}^*_2=(\overline{0},1,2,\ldots, z_1-1,\overline{z_1}, z_1-1, \ldots, 2, 1, \overline{0}),\\
%{\bf u}^*_1=(\underbrace{0,0,\ldots,0}_{z_1-1 \text{ terms}},1,2,\ldots,
%z_2-1,z_2,z_2,\ldots),\\
%{\bf u}^*_2=(\underbrace{0,0,\ldots,0}_{z_2-1 \text{ terms}},1,2,\ldots,
%z_1-1,z_1,z_1,\ldots,z_1),\\
{\bf u}^*_3=(0,0,\ldots,0)
\end{gather*}
and for $z_3<z_1+z_2$ they are
\begin{gather*}
{\bf u}^*_0=(2,3,\ldots,z_3-1,z_3,z_3,\ldots,),\\
{\bf u}^*_1=(\overline{0}, 1, 2, \ldots, z_3-z_1-1, \overline{z_3-z_1}, z_3-z_1-1,\ldots, 2, 1, \overline{0}),\\
{\bf u}^*_2=(\overline{0}, 1 , 2, \ldots, z_3-z_2-1, \overline{z_3-z_2}, z_3-z_2-1, \ldots, 2, 1, \overline{0}),\\
%{\bf u}^*_1=(\underbrace{0,0,\ldots,0}_{z_1-1 \text{ terms}},
%1,2,\ldots,z_3-z_1-1,z_3-z_1,z_3-z_1,\ldots),\\
%{\bf u}^*_2=(\underbrace{0,0,\ldots,0}_{z_2-1 \text{ terms}},
%1,2,\ldots,z_3-z_2-1,z_3-z_2,z_3-z_2,\ldots),\\
{\bf u}^*_3=(0,0,\ldots,0)\,,
\end{gather*}
where, in both cases, the total number of zeros in ${\bf u}^*_1$ is $2(z_1-1)$
and in ${\bf u}^*_2$ is $2(z_2-1)$. Their first elements are respectively $u_{1,0}=1$
and $u_{\ell,0}=0$ for $\ell=1,2$.

The above equations \eqref{eq:AKP_red} and \eqref{eq:BKP_red}
are of the form \eqref{eq:second_form} with $\psi=z_1\tilde{\psi}(z_2,z_3,\ldots,z_k)$
and $k=2,3$ respectively. More explicitly the equation \eqref{eq:AKP_red} is written as
\begin{equation}
\label{eq1}
x_n=x_M={\bf x}^{{\bf u}_1-{\bf u}_0}(-\frac{B}{A}-\frac{C}{A}{\bf x}^{{\bf u}_2-{\bf u}_1})
\end{equation}
and the equation \eqref{eq:BKP_red} as
\begin{equation}
\label{eq2}
x_n=x_N={\bf x}^{{\bf u}_1-{\bf u}_0}(-\frac{A}{D}-\frac{B}{D}{\bf x}^{{\bf u}_2-{\bf u}_1}-\frac{C}{D}{\bf x}^{{\bf u}_3-{\bf u}_1})
\end{equation}
The vectors ${\bf r}_\ell$ are related to the vectors ${\bf u}_\ell$ by
${\bf r}_1={\bf u}_1-{\bf u}_0,\, {\bf r}_\ell={\bf u}_\ell-{\bf u}_1$.

Applying Theorem \ref{prop:pois_preserv} we get the following result
about the AKP reductions.
\begin{prop}
If $z_2+z_3-z_1-2$ is even then there is a quadratic Poisson structure
of the form \eqref{eq:lv_ps} that is preserved by the map \eqref{eq:AKP_red}.
\end{prop}

We now look at some specific choices of $z_1, z_2$ and $z_3$.
\begin{prop}
\label{prop:LVPB_for_12n_even_order}
For each $n\in\mathbb N$ even with $n\geq 2$, the $n$-th order map \eqref{eq:AKP_red}
with $z_1=1, z_2=2$ and $z_3=n+1$ preserves the $\LVPS({\bf t})$ with
${\bf t}=(\overline{1,-1},1)$.
\end{prop}
\begin{proof}
For these choices of $z_1, z_2$ and $z_3$ the map \eqref{eq:AKP_red} becomes
$$
Ax_0{\bf x}^{{\bf u}_0^*}x_n+B{\bf x}^{{\bf u}_1^*}+C=0
$$
with ${\bf u}_0^*=(\bar{2})$ and ${\bf u}_1^*=(\bar{1})$.
The vectors ${\bf r}_1, {\bf r}_2$ are respectively $(\overline{-1})$
and $(0, \overline{-1})$ and the solution of the linear system \eqref{eq:linear_system_2_inv_matrix}
is ${\bf t}=(T_1,T_2,\ldots,T_{n-1})=(1,\overline{-1,1})$.
%(see also item (3) of Proposition \ref{prop:bogo_lv}).
\end{proof}
Note that the $\LVPS({\bf t})$ Poisson structure with ${\bf t}=(1,\overline{-1,1})$
is non-degenerate and the vector ${\bf t}$ is symmetric for any even $n$.
We show in the next proposition that this is the only family of the AKP reductions
that preserves a non-degenerate Poisson structure of the form \eqref{eq:lv_ps}
with symmetric vector ${\bf t}$. For the BKP
reductions we show that they cannot preserve a non-degenerate
Poisson structure of the form \eqref{eq:lv_ps}.

\begin{prop}\
\begin{enumerate}
\item The only AKP reductions \eqref{eq:AKP_red} that preserve a non-degenerate Poisson structure of the form
\eqref{eq:lv_ps} with ${\bf t}=(T_1,T_2,\ldots,T_{n-1})$ symmetric,
are those corresponding to $z_1=1, z_2=2$ and $z_3=n+1>2$
for $n$ even given in Proposition \ref{prop:LVPB_for_12n_even_order}.
\item For any choice of $z_1<z_2<z_3$, the BKP reduction \eqref{eq:BKP_red} does not
preserve a non-degenerate Poisson structure of the form \eqref{eq:lv_ps}.
\end{enumerate}
\end{prop}
\begin{proof}
The proof of item (2) is a consequence of item (3) of Proposition \ref{prop:pois_preserv_inv}
by noticing that the vectors ${\bf r}_2, {\bf r}_3$ (or equivalently the vectors ${\bf u}_1, {\bf u}_2$)
are linearly independent.

For the proof of item (1) notice that because $r_{0,1}=-1$ the solution of the
linear system \eqref{eq:linear_system_2_inv_matrix} (using the assumption that ${\bf t}$ is symmetric)
would imply that the vectors ${\bf r}^*_1, {\bf r}^*_2$ (or equivalently the vectors ${\bf u}_0, {\bf u}_1$)
are null vectors of the matrix $T^*$.
Assuming that $T$ is of full rank, the matrix $T^*$ is of co-rank $1$ and therefore
${\bf u}_0$ and ${\bf u}_1$ are linearly dependent. We can see from the explicit
formulas of ${\bf u}_0$ and ${\bf u}_1$ that they are linearly depended
if and only if $z_1=1,\, z_2=2$ and $z_3=n+1$ with $n$ even.
\end{proof}

For the AKP reduction
with $z_1=1, z_2=2$ and $z_3=n+1$ with $n$ odd we have a map of odd order for which the associated
linear system \eqref{eq:linear_system_2} does not have a non-trivial solution. This is because
\eqref{eq:linear_system_2} now has $n-1$ equations with the same number of variables and from Lemma
\ref{lem:simple_form_palindr} (item 3) we see that, for each $j=1,2,\ldots, n-1$,
there is an equation with exactly $j$ zeros.
Therefore it can be transformed into
a triangular homogeneous system with non-zero diagonal elements. As it turns out there is
a further reduction which gives rise to Poisson maps. These reductions are similar to the reductions
given in Proposition \ref{prop:reductions}. We first prove a more general result.

\begin{prop}
Let $M$ be the map \eqref{eq:map_gener}-\eqref{eq:def_of_M} with
$F=z_1\tilde{\psi}(z_2,z_3,\ldots,z_k), n$ odd, $r_{1,0}=-1, r_{\ell,0}=0$
for $\ell=2,3,\ldots, k$ and ${\bf r}^*_{\ell}$ symmetric for all $\ell=1,2,3,\ldots,k$.
Then the reduction $w_j=x_jx_{j+1}, \; j=0,1,\ldots, n-1$ of the map $M$ gives rise
to a map of order $n-1$, which is of the form \eqref{eq:second_form}.
\end{prop}
\begin{proof}
It is enough to show that under our hypotheses the equation
$x_{n-1}x_n=\frac{x_1x_{n-1}z_1}{x_1}\tilde{\psi}(z_2,z_3,\ldots,z_k)$
can be written in terms of the new variables $w_i$. For this,
it is enough to show that, when $n$ is even and the vector
${\bf r}=(r_1,r_2,\ldots,r_n)$ is symmetric, then
${\bf x}^{\bf r}=\prod_{i=1}^n x_i^{r_i}$ can be written in terms
of the $w_i$'s. This is possible if and only if
$$
\prod_{i=1}^n x_i^{r_i}=
w_1^{r_1}w_2^{r_2-r_1}w_3^{r_3-r_2+r_1}\cdots w_{n-1}^{r_{n-1}-r_{n-2}+\ldots-r_2+r_1}
$$
which is equivalent to $r_n=r_{n-1}-r_{n-2}+\ldots-r_2+r_1$.
This follows from the symmetry of the vector $\bf r$.
\end{proof}
Notice that the new exponents in the previous proposition remain symmetric.
%Note also that the previous proposition applies to both maps of the form
%\eqref{eq:first_form} and \eqref{eq:second_form}. For the maps of the form
%\eqref{eq:first_form} we shall apply the reduction $y_j=x_j+x_{j+1}$.
%
Using the previous result and Theorem \ref{prop:pois_preserv} we get the following.
\begin{prop}
\label{prop:redu}
For $z_1,z_2,z_3$ such that $n=z_2+z_3-z_1-2$ is odd, the reduction $w_j=x_jx_{j+1}, \; j=0,1,\ldots, n-1$,
of the $n$-th order map \eqref{eq:AKP_red} gives rise to an $n-1$-st order map which preserves
a quadratic Poisson structure.
\end{prop}

As a special case of the previous proposition we get the following.

\begin{cor}
For each $n\in\mathbb N$ odd with $n\geq 3$ the reduction $w_j=x_jx_{j+1}, \; j=0,1,\ldots, n-1$
of the $n$-th order map \eqref{eq:AKP_red} with $z_1=1, z_2=2$ and $z_3=n+1$
is an $n-1$-th order mapping which preserves the $\LVPS({\bf t})$ with
${\bf t}=(1,\bar{0})$.
\end{cor}
\begin{proof}
For the proof we only have to solve the associated linear system \eqref{eq:linear_system_2}.
The new map in the $w$ variables is given by
$$
Aw_0{\bf w}^{{\bf v}_0}w_{n-1}+B{\bf w}^{{\bf v}_1}+C=0
$$
with ${\bf v}_0=(\bar{1})$ and ${\bf v}_1=(\overline{1,0},1)$
and the system \eqref{eq:linear_system_2} becomes
\begin{gather*}
\sum_{j=i+1}^{n-i-2}(r_{1,i+j}-1)T_j-T_{n-i-1}=0, \; \; i=1,2,\ldots,\frac{n-1}{2}-1,\\
 \sum_{j=i+1}^{n-i-1}T_j=0, \; i=1,2,\ldots,\frac{n-1}{2}-1.
\end{gather*}
This is a linear system with $n-3$ equations and $n-2$ variables.
Its first column is zero and its rest $n-3\times n-3$ minor is
invertible since there is, for each $j=1,2,\ldots,n-3$, a row with
exactly $n-2-j$ zeros. It is not difficult to show that its solution
is indeed $T_1=1$ and $T_i=0$ for $i=2,3,\ldots,n-2$.
\end{proof}

The $\LVPS$ which is preserved by the previous reduction is
the Kac-van Moerbeke Poisson structure (see \cite{km})
which is of maximal rank for any $n$.

We now give some examples for the inverse problem: to find the maps given the
Poisson structures. We consider a family of Poisson structures that
appeared in \cite{evr_17}. Let us denote with ${\bf v}_n^{(k)}=(v_1,\ldots,v_{n-1})$
the vector with $v_i=1$ for $i=1,2,\ldots,n-k-1$ and $v_i=-1$ for $i=n-k,\ldots,n-1$.
It was shown in \cite{evr_17} that for any $n\geq3$ and $k\in\mathbb N$ with $2k+1\leq n$
the $\LVPS({\bf v}_n^{(k)})$ is of full rank when $n$ is even
and of co-rank $1$ when $n$ is odd. We give the form of the maps that preserve
the $\LVPS({\bf v}_n^{(k)})$ for the extreme cases $k=0$ and $2k+1=n$.

\begin{prop}
\label{prop:bogo_lv}
For $n\geq3$ the Poisson structure $\LVPS({\bf v}_n^{(k)})$ is preserved by maps
of the form \eqref{eq:map_gener}-\eqref{eq:def_of_M} with
$$
F({\bf x})={\bf x}^{{\bf r}_1}\tilde{\psi}({\bf x}^{{\bf r}_2}),
$$
where $\tilde{\psi}$ is any function in $C^1(\mathbb R)$ and the ${\bf r}_1, {\bf r}_2$
are given as follows.
\begin{enumerate}
\item For $k=0$ 
\begin{align*}
{\bf r}_1&=(-1,\overline{0}), \quad {\bf r}_2=(\overline{1,-1},1), \quad \text{if $n$ is odd}, \\
{\bf r}_1&=(-1,\overline{0}), \quad {\bf r}_2=(0,\overline{1,-1},1), \quad \text{if $n$ is even.}
\end{align*}
\item For $2k+1=n$
\begin{gather*}
{\bf r}_1=(1,\overline{0}), \quad {\bf r}_2=(\overline{1}).
\end{gather*}
\end{enumerate}
\end{prop}

\section*{Acknowledgements}
This research has been supported by the Australian Research Council (ARC), by
the Australian Mathematical Sciences Institute (AMSI), and by La Trobe
University (LTU). JAGR would like to acknowledge the generous hospitality
of the Department of Mathematics at LTU.

%\bibliographystyle{plain}
%\bibliography{ref}

%==============================================================
%==============================================================

%==============================================================
%==============================================================

\end{document}